\documentclass[11pt]{article}
\usepackage[margin=1in]{geometry}
\usepackage{amsmath,amssymb,amsthm,mathtools}
\usepackage{booktabs,array}
\usepackage{algorithm}
\usepackage{algpseudocode}
\usepackage{enumitem}
\usepackage{microtype}
\usepackage[hidelinks]{hyperref}

\newtheorem{theorem}{Theorem}
\newtheorem{lemma}[theorem]{Lemma}
\newtheorem{proposition}[theorem]{Proposition}
\newtheorem{corollary}[theorem]{Corollary}
\newcommand{\OPT}{\mathrm{OPT}}
\newcommand{\ALG}{\mathrm{ALG}}
\newcommand{\supp}{\operatorname{supp}}

\hypersetup{
  pdftitle={Cardinality-Constrained Randomized Assortments with Balanced Market Share},
  pdfauthor={Xiaotie Deng, Hanyu Li, Chenghua Liu}
}

\title{Cardinality-Constrained Randomized Assortments\\with Balanced Market Share}
\author{Xiaotie Deng\textsuperscript{1,2}\thanks{\href{mailto:csdeng@cityu.edu.hk}{\texttt{csdeng@cityu.edu.hk}}}\quad
Hanyu Li\textsuperscript{2}\thanks{\href{mailto:lhydave@pku.edu.cn}{\texttt{lhydave@pku.edu.cn}}}\quad
Chenghua Liu\textsuperscript{3}\thanks{\href{mailto:liuch.russell@gmail.com}{\texttt{liuch.russell@gmail.com}}}\\[0.5em]
\textsuperscript{1}Department of Computer Science, City University of Hong Kong\\[-0.1em]
\textsuperscript{2}CFCS, School of Computer Science, Peking University\\[-0.1em]
\textsuperscript{3}Institute of Software, Chinese Academy of Sciences}
\date{}

\begin{document}
\maketitle

\begin{abstract}
Assortment optimization asks a seller which of $n$ products to display to each customer. Under the multinomial-logit (MNL) model, a revenue-maximizing policy may concentrate most purchases on only a few products. Balanced market share (BMS) limits this disparity by requiring every product to receive either zero aggregate sales or at least an $\alpha$-fraction of the largest product's sales. We study randomized BMS under a per-assortment cardinality constraint: every displayed assortment contains at most $K$ products, although the seller may rotate a larger catalog across customers. This couples per-realization feasibility with aggregate balance over an exponential policy space. We resolve the model on all three algorithmic fronts. First, reweighting assortments by their MNL denominators turns the cardinality restriction into a single uniform-matroid rank inequality, yielding an exact compact sales-space formulation in which every feasible sales vector can be implemented by a policy supported on $O(n)$ assortments. Second, deciding whether there exists a feasible policy attaining a given revenue target is NP-complete even for $\alpha=1$ and $K=2$, whereas optimization is polynomial-time solvable once the set of products with positive sales is prescribed. Third, despite this hardness, BMS confines all positive sales to one multiplicative band, and a geometric scale search combined with a standard multiple-choice-knapsack dynamic program yields an FPTAS. For every $\varepsilon\in(0,1)$, it achieves at least a $(1-\varepsilon)$ fraction of the optimal expected revenue, preserves BMS and per-assortment cardinality exactly, and outputs a rational policy supported on $O(n)$ assortments.
\end{abstract}

\section{Introduction}\label{sec:intro}

Assortment optimization asks a seller to choose which products to offer to a customer.  It is a basic problem in revenue management and appears whenever an online retailer, marketplace, or advertising platform selects a menu from a larger catalog.  We consider the standard multinomial-logit (MNL) model with product set $N=[n]=\{1,\ldots,n\}$.  Product $i$ has revenue $r_i>0$ and attractiveness $v_i>0$.  If a customer is shown an assortment $S\subseteq N$, then she purchases product $i\in S$ with probability
\[
 \phi_i(S)=\frac{v_i}{1+\sum_{j\in S}v_j};
\]
the term $1$ represents the no-purchase option.  The seller may vary the displayed assortment across customers, so a policy is a distribution $q=(q_S)$ over assortments.  Its aggregate purchase probability and expected revenue are
\[
 x_i=\sum_Sq_S\phi_i(S),
 \qquad
 R(q)=\sum_i r_ix_i.
\]

Optimizing revenue alone can create highly uneven sales.  A few products may capture most of the demand, while other products that remain active sell only rarely.  El Housni, Feng, and Topaloglu~\cite{EFTopaloglu} introduced \emph{balanced market share} (BMS) to control this disparity through one relative-balance parameter.  A randomized policy is BMS-feasible when
\[
 x_i=0
 \quad\text{or}\quad
 x_i\ge\alpha\max_jx_j
 \qquad\text{for every }i.\qquad\text{(BMS)}
\]
where $\alpha\in(0,1]$.  The seller remains free to decide which products are active.  Once a product has positive sales, however, its purchase probability must be at least an $\alpha$ fraction of the largest one.  Thus $\alpha=1$ requires equal aggregate sales among active products, while smaller values permit more disparity.  The constraint concerns purchases aggregated across customers; it does not prescribe display frequencies or require every product to participate.

In the unrestricted static model of~\cite{EFTopaloglu}, the seller may randomize over arbitrary assortments.  The authors prove that this problem is exactly solvable in polynomial time.  Some optimal active set is characterized by one threshold on revenue and another on attractiveness, and the corresponding aggregate sales can be implemented by a distribution over at most $n$ nested assortments.  They also study additional restrictions on the \emph{active catalog}
\[
 A(q)=\bigcup_{S:q_S>0}S,
\]
the set of products that appear with positive probability anywhere in the policy.

Operational constraints often apply to each display instead.  A recommendation panel or physical shelf may contain at most $K$ products for any one customer, while the seller can rotate a larger collection over time.  The resulting restriction is
\[
 q_S>0\quad\Longrightarrow\quad |S|\le K.\qquad\text{(C)}
\]
This condition is strictly weaker than $|A(q)|\le K$.  For example, when $K=2$, a seller may randomize over $\{1,2\}$, $\{2,3\}$, and $\{3,4\}$.  Every customer sees at most two products, although four products participate.  El Housni, Feng, and Topaloglu explicitly left BMS with cardinality or capacity constraints on every randomized assortment as an open direction~\cite{EFTopaloglu}.

This distinction changes both the geometry and the complexity of the problem.  There are exponentially many assortments satisfying (C), while BMS is enforced only after their purchase probabilities are aggregated.  The sales vector $x$ does not by itself reveal whether it admits an implementation using only small assortments, and the induced rank constraint need not preserve the two-threshold geometry of unrestricted BMS.  We resolve the cardinality-constrained open direction on all three algorithmic fronts: we characterize attainable aggregate sales exactly, establish the exact complexity of revenue maximization, and construct an implementable fully polynomial-time approximation scheme (FPTAS) despite NP-hardness.

We give a complete resolution.  The first theorem shows that the exponential policy space has an exact compact image.

\begin{theorem}[Exact sales-space formulation]\label{thm:sales}
The cardinality-constrained randomized BMS problem is equivalent, preserving feasibility and objective value, to the following problem in rescaled-sales variables $w\in\mathbb R_+^n$:
\begin{align}
 \max\quad &F(w):=\frac{\sum_i r_iw_i}{1+\sum_iw_i} \label{eq:obj}\\
 \text{s.t.}\quad &0\le w_i\le v_i &&(i\in N),\label{eq:caps}\\
 &\sum_i\frac{w_i}{v_i}\le K,\label{eq:rank}\\
 &w_i=0\ \text{or}\ w_i\ge\alpha\max_jw_j &&(i\in N).\label{eq:bmsw}
\end{align}
Every feasible rational $w$ can be implemented by a rational distribution over $O(n)$ assortments of size at most $K$.
\end{theorem}

Compact representability does not preserve the exact tractability of unrestricted BMS.

\begin{theorem}[Exact complexity]\label{thm:hardness}
Given a rational revenue target $H$, deciding whether there exists a feasible policy with expected revenue at least $H$ is NP-complete, even when $\alpha=1$ and $K=2$ are fixed.
\end{theorem}

Despite this hardness, the remaining structure supports arbitrarily accurate approximation without relaxing either feasibility requirement.

\begin{theorem}[Exact-feasibility FPTAS]\label{thm:fptas}
For every rational $\varepsilon\in(0,1)$, there is an algorithm polynomial in the input length and $1/\varepsilon$ that returns at least a $(1-\varepsilon)$ fraction of the optimal expected revenue, satisfies BMS and per-assortment cardinality exactly, and outputs a rational distribution supported on $O(n)$ assortments.
\end{theorem}

The proof of Theorem~\ref{thm:sales} uses a denominator tilt. Reweighting each assortment by the inverse of its MNL denominator produces a tilted distribution $p=(p_S)$ over assortments under which the rescaled sales are inclusion marginals:
\[
 \frac{w_i}{v_i}=\sum_{S:i\in S}p_S=\Pr_{S\sim p}(i\in S).
\]
Because every realized assortment has size at most $K$, these marginals lie in the independence polytope of a rank-$K$ uniform matroid. Conversely, every point of this polytope is a convex combination of sets of size at most $K$. A circular-interval decomposition uses only $O(n)$ such sets, and reversing the tilt recovers an MNL policy with the prescribed aggregate sales. This is the exact bridge from the exponential policy space to the compact formulation.

The compact formulation yields a sharp tractability boundary. Once the active set $A=\supp(w):=\{i\in N:w_i>0\}$, equivalently $A(q)$, is fixed, the BMS disjunction becomes linear, and a Charnes--Cooper transformation solves the remaining problem in polynomial time. The unrestricted model handles this combinatorial choice through a two-threshold characterization~\cite{EFTopaloglu}, but the cardinality-induced rank constraint destroys that geometry: even a three-product instance can have a unique optimal support that is not described by revenue and attractiveness thresholds. Our reduction from Subset Sum~\cite{GareyJohnson} proves that optimizing over $A$ is NP-hard already for $\alpha=1$ and $K=2$, and that the associated revenue-threshold decision problem is NP-complete. Balance forces all selected coordinates to share one scale, while the rank and revenue conditions encode the subset-sum target. Thus the computational hardness is concentrated in endogenous support selection: representation and prescribed-active-set optimization remain polynomial-time solvable.

The FPTAS exploits the same common-scale structure. In every nonzero BMS-feasible solution, all positive coordinates lie in one band $[\tau,\tau/\alpha]$, where $\tau:=\alpha\max_iw_i>0$. We search a geometric grid for this scale and discretize the coordinate values within the corresponding band. For a candidate revenue $\rho$, the fractional objective linearizes to $\sum_i(r_i-\rho)w_i\ge\rho$; after the scale is fixed, selecting zero or one value for each product becomes a multiple-choice knapsack problem under the rank budget. Standard profit scaling gives a polynomial dynamic program. All selected positive values remain in a common balance band, the rank budget is enforced exactly, and Theorem~\ref{thm:sales} reconstructs a legal policy supported on $O(n)$ assortments. Approximation therefore affects only revenue, not BMS or per-assortment cardinality.

This paper is organized as follows. Section~\ref{sec:related} reviews related work, and Section~\ref{sec:model} gives the formal model. Section~\ref{sec:sales} proves the exact compact formulation and $O(n)$-support reconstruction theorem. Section~\ref{sec:exact} establishes the sharp computational frontier: prescribed-support optimization is polynomial-time solvable, the unrestricted two-threshold structure fails under the rank constraint, and the revenue-threshold decision problem is NP-complete even for $\alpha=1$ and $K=2$. Section~\ref{sec:fptas} develops the exact-feasibility FPTAS, and Section~\ref{sec:conclusion} concludes. The certificate for the threshold counterexample appears after the references.

\section{Related work}\label{sec:related}

El Housni, Feng, and Topaloglu introduce BMS, solve its unrestricted static version, and explicitly distinguish restrictions on the aggregate active catalog from restrictions on every randomized assortment~\cite{EFTopaloglu}.  Our problem is the cardinality-constrained case of the latter open model, and we provide its complete algorithmic resolution.  Chen, Golrezaei, and Susan study randomized cardinality-constrained MNL assortments under exogenous pairwise-difference constraints on item outcomes~\cite{ChenGolrezaeiSusan}.  Lu, Sahin, and Wang study randomized fair assortment planning with minimum-exposure constraints and also consider cardinality restrictions on the assortments~\cite{LuSahinWang}.  Zhu, Rusmevichientong, and Topaloglu allow a product class to be inactive or require its aggregate market share to meet a class-specific absolute threshold~\cite{ZhuRusTop}.  In contrast, our BMS constraint is product-level and endogenous: every positive market share is compared with the largest positive market share.

Ordinary cardinality-constrained MNL assortment optimization is polynomial-time solvable~\cite{RusShenShmoys}.  Cardinality combined with visibility requirements can make related assortment problems strongly NP-hard~\cite{BarreEtAl}, although that model and its approximation boundary differ from aggregate relative balance.  More broadly, assortment optimization has been studied under general regular discrete-choice models~\cite{BerbegliaJoret} and under combinatorial bundle-choice valuation models~\cite{ImmorlicaEtAl}.

The individual ingredients have classical antecedents: the MNL sales-space viewpoint underlies unrestricted BMS~\cite{EFTopaloglu}, the marginal polytope in Theorem~\ref{thm:sales} is the uniform-matroid independence polytope~\cite{Schrijver}, prescribed-support optimization uses linear-fractional programming~\cite{CharnesCooper}, and the FPTAS uses profit scaling for multiple-choice knapsack~\cite{KellererEtAl}.  Our contribution is their exact synthesis. Under the denominator tilt, per-assortment cardinality becomes a single uniform-matroid rank inequality; the remaining endogenous support choice makes the revenue-threshold decision problem NP-complete, yet BMS retains a common-scale geometry strong enough to yield an exact-feasibility FPTAS and an $O(n)$-support implementation.

\section{Model and optimization problem}\label{sec:model}
We use the standard MNL assortment model~\cite{RusShenShmoys,TalluriVanRyzin}. There are $n$ products, collected in $N=[n]=\{1,\ldots,n\}$. Product $i$ has revenue $r_i>0$ and attractiveness weight $v_i>0$. The no-purchase option, indexed by $0$, has revenue zero and weight $v_0=1$. When the seller displays an assortment $S\subseteq N$, write
\[
 d(S)=1+\sum_{i\in S}v_i.
\]
The customer purchases product $i$ or chooses the no-purchase option with probabilities
\[
 \phi_i(S)=\frac{v_i\mathbf 1[i\in S]}{d(S)},\qquad
 \phi_0(S)=\frac1{d(S)}.
\]
Here $\mathbf 1[\cdot]$ denotes the indicator of an event.

Fix a display limit $K\in[n]$, and let $\mathcal S_K=\{S\subseteq N:|S|\le K\}$. A randomized policy is a distribution $q=(q_S)_{S\in\mathcal S_K}$, where $q_S\ge0$ and $\sum_{S\in\mathcal S_K}q_S=1$. For each customer, the seller draws $S\sim q$ and displays that assortment. Thus the cardinality bound applies to every realized display, while the union of products used across displays may be larger than $K$.

The policy induces aggregate purchase probabilities
\[
 x_i=x_i(q)=\sum_{S\in\mathcal S_K}q_S\phi_i(S)\quad(i\in N),
 \qquad
 x_0=\sum_{S\in\mathcal S_K}q_S\phi_0(S),
\]
with $x_0+\sum_i x_i=1$. Its expected revenue is
\[
 R(q)=\sum_{i=1}^n r_i x_i.
\]
This aggregate-sales viewpoint follows the sales-based formulations in~\cite{GallegoRatliffShebalov,Topaloglu}.

Following~\cite{EFTopaloglu}, fix a balance parameter $\alpha\in(0,1]$. The active set is
\[
 A(q)=\{i\in N:x_i(q)>0\}
     =\bigcup_{S:q_S>0}S.
\]
A policy is \emph{BMS-feasible} if every active product receives at least an $\alpha$ fraction of the largest aggregate purchase probability, equivalently,
\[
 x_i=0\quad\text{or}\quad
 x_i\ge \alpha\max_{j\in N}x_j
 \qquad(i\in N).
\]
The optimization problem is to compute
\[
 \OPT=\max\{R(q):q\text{ is a distribution on }\mathcal S_K
                         \text{ and is BMS-feasible}\}.
\]

For computational statements, $r_i$, $v_i$, and $\alpha$ are binary-encoded positive rationals and $K$ is an integer. The decision version additionally receives a rational threshold $H$ and asks whether $\OPT\ge H$. An output policy is represented by listing its positive-support assortments and their rational probabilities; the approximation algorithm must satisfy both BMS and the display limit exactly.

\section{An exact attainable-sales polytope}\label{sec:sales}
We prove Theorem~\ref{thm:sales} by first mapping a policy to compact sales variables and then reversing the map. Start with a feasible policy $q$ and define
\[
 x_0=\sum_S\frac{q_S}{d(S)},\qquad w_i=\frac{x_i}{x_0},\qquad p_S=\frac{q_S}{d(S)x_0}.
\]
Then $\sum_Sp_S=1$, and
\[
 z_i:=\frac{w_i}{v_i}=\sum_{S:i\in S}p_S=\Pr_{S\sim p}(i\in S).
\]
Thus $z$ belongs to the independence polytope of the rank-$K$ uniform matroid:
\[
 0\le z_i\le1,\qquad \sum_i z_i\le K.
\]
This gives \eqref{eq:caps} and \eqref{eq:rank}. Since
\[
 x_0=\frac1{1+\sum_iw_i},\qquad x_i=\frac{w_i}{1+\sum_jw_j},
\]
BMS is equivalent to \eqref{eq:bmsw}, and revenue becomes \eqref{eq:obj}.

For the reverse direction, let $w$ satisfy \eqref{eq:caps} and \eqref{eq:rank}, put $z_i=w_i/v_i$, and express $z$ as a convex combination
\[
 z=\sum_Sp_S\mathbf 1_S,\qquad |S|\le K,
\]
where $\mathbf 1_S$ denotes the incidence vector of $S$.
Define
\[
 q_S=\frac{p_S(1+\sum_{i\in S}v_i)}{1+\sum_iw_i}.
\]
Then $\sum_Sq_S=1$, and direct substitution recovers the stated aggregate probabilities.

It remains to obtain the stated support bound constructively. Add at most $K$ dummy coordinates so that the augmented vector has sum $K$, and arrange intervals of lengths $z_i$ on $[0,K)$. For $U$ uniform on $[0,1)$, select the labels covering $U,U+1,\ldots,U+K-1$, discarding dummy labels. Each label is selected with marginal probability $z_i$, labels are distinct, and the selected set has size at most $K$. The set changes only at fractional parts of cumulative interval endpoints, yielding $O(n)$ rational support points. Together with the reverse tilt above, this proves both the exact equivalence and the implementation guarantee.

\section{Exact optimization reduces to support selection}\label{sec:exact}
The compact formulation removes the exponential policy space but leaves the endogenous choice of which products receive positive sales. We first show that optimization is polynomial-time solvable once the exact active set $\supp(w)$ is prescribed. We then turn to the endogenous-support barrier: the threshold structure of unrestricted BMS does not survive the rank constraint, optimizing over the active set is NP-hard, and the associated revenue-threshold decision problem is NP-complete.

\subsection{A prescribed support is tractable}\label{sec:fixedsupport}

\begin{proposition}[Prescribed-support optimization]\label{prop:fixedsupport}
For every prescribed nonempty active set $A\subseteq N$, the maximum revenue among feasible vectors with $\supp(w)=A$ can be computed in polynomial time.
\end{proposition}

\begin{proof}
Set $w_i=0$ for $i\notin A$. On $A$, balanced market share is represented by the linear inequalities
\[
 w_i\ge \alpha w_j\qquad(i,j\in A).
\]
Together with
\[
 0\le w_i\le v_i,\qquad \sum_{i\in A}\frac{w_i}{v_i}\le K,
\]
these define a bounded rational polytope. The objective is linear-fractional. Introduce the Charnes--Cooper variables
\[
 s=\frac1{1+\sum_{i\in A}w_i},\qquad y_i=sw_i.
\]
The problem becomes the rational linear program
\begin{align*}
 \max\quad &\sum_{i\in A}r_iy_i\\
 \text{s.t.}\quad &s+\sum_{i\in A}y_i=1,\\
 &0\le y_i\le v_is &&(i\in A),\\
 &\sum_{i\in A}\frac{y_i}{v_i}\le Ks,\\
 &y_i\ge \alpha y_j &&(i,j\in A),\\
 &s\ge0.
\end{align*}
This linear program has polynomial size and rational data, so it is solvable in polynomial time. A sufficiently small vector that is positive on every coordinate in $A$ is feasible and has positive revenue. Hence an optimum of the linear program is nonzero, and the pairwise inequalities force every coordinate in $A$ to be positive. The linear program therefore solves the exact-active-set problem.
\end{proof}

\begin{corollary}[Closed-form fixed-support value for $\alpha=1$]\label{cor:alpha1fixed}
Suppose $\alpha=1$, and fix a nonempty active set $A$. Then all positive coordinates are equal, $w_i=t$ for $i\in A$. The largest feasible common value is
\[
 t_A=\min\left\{\min_{i\in A}v_i,\ \frac{K}{\sum_{i\in A}1/v_i}\right\},
\]
and the optimal revenue supported on $A$ is
\[
 F_A=\frac{t_A\sum_{i\in A}r_i}{1+|A|t_A}.
\]
\end{corollary}

\begin{proof}
When $\alpha=1$, the BMS inequalities imply $w_i=w_j$ for all $i,j\in A$. Writing their common value as $t$, the coordinate caps require $t\le\min_{i\in A}v_i$, while the rank constraint requires
\[
 t\sum_{i\in A}\frac1{v_i}\le K.
\]
Thus $t\le t_A$. For fixed $A$,
\[
 f_A(t)=\frac{t\sum_{i\in A}r_i}{1+|A|t},\qquad
 f_A'(t)=\frac{\sum_{i\in A}r_i}{(1+|A|t)^2}>0,
\]
so the largest feasible $t$ is optimal.
\end{proof}

\subsection{Threshold failure and NP-completeness}\label{sec:hardness}
The unrestricted BMS model admits an optimal active set described by one revenue threshold and one attractiveness threshold~\cite{EFTopaloglu}. In the $(r_i,v_i)$ plane, such a set is a northeast rectangle, namely a set of the form $\{i:r_i\ge\bar r,\ v_i\ge\bar v\}$ for some thresholds $\bar r$ and $\bar v$. The rank budget $\sum_iw_i/v_i\le K$ couples the products and destroys this geometry.

\begin{proposition}[Nonrectangular optimal support]\label{prop:nonrectangular}
There is a rational three-product instance with $K=2$ whose unique optimal active set is not a northeast rectangle in revenue and attractiveness.
\end{proposition}

Consider $\alpha=1/6$, $K=2$, and
\[
\begin{array}{c|ccc}
 i&1&2&3\\\hline
 r_i&65&80&64\\
 v_i&3&2&14
\end{array}
\]
The unique optimum has active set $\{2,3\}$, with $w_2=2$, $w_3=12$, and value $928/15$. Any northeast rectangle containing products 2 and 3 must also contain product 1. Appendix~\ref{app:counterexample} gives the exact values of all competing supports and a dual certificate for the full-support optimum.

The counterexample above rules out the threshold shortcut available in the unrestricted model. We now show that the remaining combinatorial choice of $A$ is NP-hard even when balance and cardinality are fixed at their simplest nontrivial values.
\begin{proof}[Proof of Theorem~\ref{thm:hardness}]
The case $H\le0$ is trivial. For $H>0$, membership in NP follows by guessing the active set $A$, imposing $w_i=0$ for $i\notin A$, the pairwise BMS inequalities on $A$, and
\[
 \sum_i(r_i-H)w_i\ge H.
\]
This is a bounded rational linear system and, if feasible, has a polynomial-bit vertex certificate.

For hardness, take a positive-integer Subset Sum instance with items $a_1,\ldots,a_m\in\mathbb Z_{>0}$ and target $B\in\mathbb Z_{>0}$, asking whether some subset sums to $B$. Delete every item $a_i>B$. If no item remains, output the fixed legal no-instance
\[
 n=2,\quad r_1=r_2=v_1=v_2=1,\quad \alpha=1,\quad K=2,\quad H=2.
\]
Otherwise set
\[
 \alpha=1,\qquad K=2,\qquad H=3B.
\]
Create an anchor product $\star$ and one product per remaining integer:
\[
 v_\star=1,\quad r_\star=5B,\qquad v_i=\frac{B}{a_i},\quad r_i=3B+a_i.
\]
When $\alpha=1$, all positive $w_i$'s equal a common value $t$. Let $T$ be the active item set, write $a(T):=\sum_{i\in T}a_i$, and let $\delta_\star\in\{0,1\}$ indicate whether the anchor is active. Revenue at least $3B$ is equivalent to
\begin{equation*}
 t(2B\delta_\star+a(T))\ge3B.\qquad(\ast)
\end{equation*}
Without the anchor, \eqref{eq:rank} implies $t a(T)/B\le2$, so the left-hand side of $(\ast)$ is at most $2B$.

With the anchor, $t\le1$ and
\[
 t(1+a(T)/B)\le2,\qquad t\le\min\{1,2B/(B+a(T))\}.
\]
If $a(T)<B$, even $t=1$ fails $(\ast)$. If $a(T)=B$, $t=1$ is feasible and attains equality. If $a(T)>B$, the largest possible left-hand side is
\[
 \frac{2B(2B+a(T))}{B+a(T)}<3B.
\]
Thus the constructed instance is a yes-instance exactly when some subset sums to $B$.

For completeness, when $T$ sums to $B$, the original policy can randomize over the pairs $\{\star,i\}$ using
\[
 q_{\star,i}=\frac{B+2a_i}{B(|T|+2)}.
\]
These probabilities sum to one and yield $x_0=x_\star=x_i=1/(|T|+2)$ for $i\in T$, with revenue $3B$.
\end{proof}

Corollary~\ref{cor:alpha1fixed} makes the source of hardness especially transparent: every prescribed support has an explicit value, yet selecting the support of maximum value remains NP-hard.

\section{An exact-feasibility FPTAS}\label{sec:fptas}
Despite NP-hardness, BMS still forces every positive coordinate to lie near one common scale. We first localize and discretize this scale. We then reduce revenue testing to multiple-choice knapsack and combine the two ingredients into the FPTAS.

\subsection{Locating and discretizing the balance scale}
Fix $\varepsilon\in(0,1)$ and put $\delta=\varepsilon/10$. Let $v_{\min}=\min_i v_i$, $v_{\max}=\max_i v_i$, $r_{\max}=\max_i r_i$, and let
\[
 L=\max_i\frac{r_iv_i}{1+v_i}
\]
be the best singleton revenue.

\begin{lemma}[Scale range]\label{lem:scale}
There is an optimum $w^\star$ for which, with $\tau^\star=\alpha\max_iw_i^\star$,
\[
 \frac{\alpha v_{\min}}n\le\tau^\star\le\alpha v_{\max}.
\]
\end{lemma}
\begin{proof}
At an optimum, some cap or the rank budget is tight; otherwise all positive coordinates can be scaled up while preserving BMS and strictly increasing the fractional objective. If a cap is tight, $\max_iw_i^\star\ge v_{\min}$. If the budget is tight, then
\[
 K=\sum_iw_i^\star/v_i\le n\max_iw_i^\star/v_{\min},
\]
so $\max_iw_i^\star\ge Kv_{\min}/n\ge v_{\min}/n$. The upper bound follows from the caps.
\end{proof}

Use the scale grid
\[
 \mathcal T=\left\{\frac{\alpha v_{\min}}n(1+\delta)^j:\ j\in\mathbb Z_{\ge0},\ \frac{\alpha v_{\min}}n(1+\delta)^j\le\alpha v_{\max}\right\}.
\]
For $\tau\in\mathcal T$, product $i$ has the zero option and
\[
 \mathcal O_i(\tau)=\{\tau(1+\delta)^\ell:\ \ell\in\mathbb Z_{\ge0},\ \tau(1+\delta)^\ell\le\min\{v_i,\tau/\alpha\}\}.
\]
Every nonzero discrete coordinate lies in the same exact BMS band $[\tau,\tau/\alpha]$.

For $\tau\in\mathcal T$, define the corresponding discrete optimum by
\[
 F_D(\tau):=\max\left\{F(w):
 w_i\in\{0\}\cup\mathcal O_i(\tau)\ (i\in N),
 \sum_i\frac{w_i}{v_i}\le K\right\}.
\]

\begin{lemma}[Discretization]\label{lem:disc}
For some $\tau\in\mathcal T$, the discrete optimum satisfies
\[
 F_D(\tau)\ge\frac{\OPT}{(1+\delta)^2}.
\]
\end{lemma}
\begin{proof}
Take the largest grid point $\tau\le\tau^\star$. First scale $w^\star$ by $\tau/\tau^\star$, losing at most a factor $1+\delta$. Then round every positive coordinate down to the largest option in $\mathcal O_i(\tau)$, losing less than another factor $1+\delta$. Caps and the rank budget only become easier, and all positive coordinates remain in the exact common BMS band.
\end{proof}

\subsection{Revenue testing and the FPTAS}
We use the revenue grid
\[
 \mathcal P=\{L(1+\delta)^j:j\in\mathbb Z_{\ge0},\ L(1+\delta)^j\le r_{\max}\}.
\]
For fixed $\tau\in\mathcal T$ and $\rho\in\mathcal P$, a discrete vector has revenue at least $\rho$ if and only if
\begin{equation}
 \sum_i(r_i-\rho)w_i\ge\rho.\label{eq:test}
\end{equation}
Each positive option $u\in\mathcal O_i(\tau)$ has transformed profit $\pi(i,u)=(r_i-\rho)u$ and weight $c(i,u)=u/v_i$. Discard positive options with nonpositive transformed profit, while always retaining the zero option. Choosing at most one positive option for each product is then a multiple-choice knapsack (MCKP) problem with capacity $K$.

If at least one positive option remains, let $\pi_{\max}$ be its largest transformed profit, set $\theta=\delta\pi_{\max}/n$, and define $\widehat\pi(i,u)=\lfloor\pi(i,u)/\theta\rfloor$. The dynamic program (DP) stores, for each total scaled profit, the minimum exact rational weight and one predecessor.

\begin{lemma}[MCKP approximation]\label{lem:mckp}
The profit-scaling DP returns actual transformed profit at least $(1-\delta)$ times the optimum of the corresponding MCKP.
\end{lemma}
\begin{proof}
Every individual option has weight at most one and is feasible because $K\ge1$; hence the MCKP optimum $P^\star\ge\pi_{\max}$. Rounding down each selected profit loses at most $\theta$, so at most $n\theta=\delta\pi_{\max}\le\delta P^\star$ in total. The number of scaled profit states is $O(n^2/\delta)$.
\end{proof}

Combining the two grids with this routine gives the complete algorithm. Algorithm~\ref{alg:fptas} enumerates the grids only after both have been constructed. For each pair $(\tau,\rho)$, it runs the profit-scaled knapsack routine and verifies the candidate using its unrounded transformed profit.

\begin{algorithm}[H]
\small
\caption{FPTAS for cardinality-constrained randomized BMS}
\label{alg:fptas}
\begin{algorithmic}[1]
\Require Rational data $(r_i,v_i)_{i\in N}$, $\alpha\in(0,1]$, $K\in[n]$, and rational $\varepsilon\in(0,1)$
\Ensure A rational BMS-feasible policy over assortments of size at most $K$
\State Set $\delta=\varepsilon/10$; compute $L,v_{\min},v_{\max},r_{\max}$ and construct $\mathcal T,\mathcal P$.
\State Set the best singleton sales vector as the incumbent $w^{\mathrm{inc}}$.
\For{$\tau\in\mathcal T$}
  \State Construct $\{0\}\cup\mathcal O_i(\tau)$ for every product $i$.
  \For{$\rho\in\mathcal P$}
    \State Form each option group $\{0\}\cup\{u\in\mathcal O_i(\tau):\pi(i,u)>0\}$.
    \State If no positive option remains, \textbf{continue}.
    \State Let $\pi_{\max}$ be the largest transformed profit among the remaining positive options.
    \State Set $\theta=\delta\pi_{\max}/n$ and $\widehat\pi(i,u)=\lfloor\pi(i,u)/\theta\rfloor$.
    \State Run the MCKP DP over the products, storing minimum exact weight and a predecessor for each total scaled profit.
    \State Recover the feasible state with largest scaled profit and its sales vector $\widetilde w$.
    \State If its actual transformed profit is at least $\rho$, set $w^{\mathrm{inc}}\gets\widetilde w$ if its actual revenue improves the incumbent.
  \EndFor
\EndFor
\State Convert $w^{\mathrm{inc}}$ to a policy $q$ supported on $O(n)$ assortments using Theorem~\ref{thm:sales}.
\State \Return $q$.
\end{algorithmic}
\end{algorithm}

\begin{proof}[Proof of Theorem~\ref{thm:fptas}]
Let $\ALG$ denote the revenue of the policy returned by Algorithm~\ref{alg:fptas}. Take the scale from Lemma~\ref{lem:disc} and abbreviate its discrete optimum by $F_D$. Put $\bar\rho=F_D/(1+2\delta)$. If $\bar\rho<L$, the singleton incumbent is already at least $F_D/(1+2\delta)$.

Otherwise choose the largest $\rho\in\mathcal P$ with $\rho\le\bar\rho$. Then
\[
 \rho>\frac{F_D}{(1+2\delta)(1+\delta)}.
\]
If $w^D$ attains $F_D$ and $W_D=\sum_iw_i^D$, its transformed profit is
\[
 \sum_i(r_i-\rho)w_i^D=F_D+(F_D-\rho)W_D\ge F_D\ge(1+2\delta)\rho.
\]
Deleting nonpositive-profit coordinates preserves the exact common BMS band and cannot reduce this value. By Lemma~\ref{lem:mckp}, the DP returns transformed profit at least
\[
 (1-\delta)(1+2\delta)\rho\ge\rho,
\]
so the vector passes \eqref{eq:test}. Therefore
\[
 \ALG\ge\frac{\OPT}{(1+\delta)^3(1+2\delta)}\ge(1-\varepsilon)\OPT,
\]
where the last inequality follows because
\[
 3\log(1+\delta)+\log(1+2\delta)
 \le5\delta=\frac\varepsilon2,
\]
and therefore $1/((1+\delta)^3(1+2\delta))\ge e^{-\varepsilon/2}\ge1-\varepsilon$.

The scale, option, and revenue grids have polynomial length and polynomial-bit rational entries. The DP uses $O(n^2/\delta)$ profit states per grid pair and keeps every exact option weight $c(i,u)=u/v_i$. All positive returned coordinates lie in one band $[\tau,\tau/\alpha]$, so BMS is exact, and the accepted total weight is at most $K$, so the rank constraint is exact. Finally, Theorem~\ref{thm:sales} reconstructs a legal policy supported on $O(n)$ assortments. Thus approximation affects only revenue, not either feasibility requirement.
\end{proof}

\section{Conclusion}\label{sec:conclusion}
For cardinality-constrained randomized MNL assortments, an exact denominator tilt collapses the exponential policy space to a compact sales-space formulation, and every feasible rational point can be implemented by a rational policy supported on $O(n)$ assortments. This resolves the open direction on representability and yields a sharp complexity separation: optimization is polynomial-time for a prescribed active set, whereas the revenue-threshold decision problem is NP-complete---even for $\alpha=1$ and $K=2$---when the active set is endogenous.

The same structure also explains the positive result. Although the active set is combinatorial, all positive coordinates lie in a single balance band. Geometric discretization of that band, combined with a standard profit-scaled dynamic program, yields an FPTAS without relaxing either BMS or per-realization cardinality. Extending this exact-collapse viewpoint beyond the uniform matroid is a natural next question.

\section*{Acknowledgments}
OpenAI Codex assisted with proof development and verification. The authors
independently checked all arguments and assume full responsibility for the
content of this manuscript.

\clearpage
\appendix
\section{Certificate for the nonrectangular optimum}\label{app:counterexample}
For the instance in Proposition~\ref{prop:nonrectangular}, the active set $\{2,3\}$ attains
\[
 w_2=2,\qquad w_3=12,\qquad F(w)=\frac{928}{15}.
\]
It is implemented, for example, by
\[
 q_{\{2,3\}}=\frac{34}{35},\qquad q_{\{2\}}=\frac1{35}.
\]
To certify global optimality and uniqueness, the best values on singleton and two-product supports are
\[
\begin{array}{c|cccccc}
 A&\{1\}&\{2\}&\{3\}&\{1,2\}&\{1,3\}&\{2,3\}\\\hline
 \max_{\supp(w)=A}F(w)&195/4&160/3&896/15&355/6&1091/18&928/15
\end{array}
\]
For full support, set $\rho=11989/195$. To certify $F(w)\le\rho$, it suffices to prove
\[
 686w_1+3611w_2+491w_3\le11989
\]
because this is exactly $\sum_i(r_i-\rho)w_i\le\rho$ after multiplication by $195$. The displayed inequality is obtained by adding the valid inequalities
\[
 4767\left(\frac{w_1}{3}+\frac{w_2}{2}+\frac{w_3}{14}\right)\le9534,
\]
\[
 \frac{301}{2}(w_3-6w_1)\le0,\qquad \frac{2455}{2}w_2\le2455.
\]
Equality holds at
\[
 (w_1,w_2,w_3)=\left(\frac{21}{16},2,\frac{63}{8}\right),
\]
so the full-support optimum is $11989/195<928/15$. This certifies the claimed global optimum and the failure of rectangular support structure.


\begin{thebibliography}{15}
\bibitem{BarreEtAl}
Th\'eo Barr\'e, Omar El Housni, Marouane Ibn Brahim, Andrea Lodi, and Danny Segev.
Assortment optimization with visibility constraints.
\emph{Mathematical Programming}, 216(1--2):177--220, 2026.

\bibitem{BerbegliaJoret}
Gerardo Berbeglia and Gwena\"el Joret.
Assortment optimisation under a general discrete choice model: A tight analysis of revenue-ordered assortments.
\emph{Algorithmica}, 82(4):681--720, 2020.

\bibitem{CharnesCooper}
Abraham Charnes and William W. Cooper.
Programming with linear fractional functionals.
\emph{Naval Research Logistics Quarterly}, 9(3--4):181--186, 1962.

\bibitem{ChenGolrezaeiSusan}
Qinyi Chen, Negin Golrezaei, and Fransisca Susan.
Fair assortment planning.
arXiv:2208.07341, 2022.

\bibitem{EFTopaloglu}
Omar El Housni, Qing Feng, and Huseyin Topaloglu.
Optimal selection with balanced market share: Static and dynamic assortment optimization.
arXiv:2507.05606, 2025.

\bibitem{GallegoRatliffShebalov}
Guillermo Gallego, Richard Ratliff, and Sergey Shebalov.
A general attraction model and sales-based linear program for network revenue management under customer choice.
\emph{Operations Research}, 63(1):212--232, 2015.

\bibitem{GareyJohnson}
Michael R. Garey and David S. Johnson.
\emph{Computers and Intractability: A Guide to the Theory of NP-Completeness}.
W. H. Freeman, 1979.

\bibitem{ImmorlicaEtAl}
Nicole Immorlica, Brendan Lucier, Jieming Mao, Vasilis Syrgkanis, and Christos Tzamos.
Combinatorial assortment optimization.
\emph{ACM Transactions on Economics and Computation}, 9(1):1--34, 2021.

\bibitem{KellererEtAl}
Hans Kellerer, Ulrich Pferschy, and David Pisinger.
\emph{Knapsack Problems}.
Springer, 2004.

\bibitem{LuSahinWang}
Wentao Lu, \"{O}zge \c{S}ahin, and Ruxian Wang.
A simple way towards fair assortment planning: Algorithms and welfare implications.
SSRN 4514495, 2023.

\bibitem{RusShenShmoys}
Paat Rusmevichientong, Zuo-Jun Max Shen, and David B. Shmoys.
Dynamic assortment optimization with a multinomial logit choice model and capacity constraint.
\emph{Operations Research}, 58(6):1666--1680, 2010.

\bibitem{Schrijver}
Alexander Schrijver.
\emph{Combinatorial Optimization: Polyhedra and Efficiency}.
Springer, 2003.

\bibitem{TalluriVanRyzin}
Kalyan Talluri and Garrett van Ryzin.
Revenue management under a general discrete choice model of consumer behavior.
\emph{Management Science}, 50(1):15--33, 2004.

\bibitem{Topaloglu}
Huseyin Topaloglu.
Joint stocking and product offer decisions under the multinomial logit model.
\emph{Production and Operations Management}, 22(5):1182--1199, 2013.

\bibitem{ZhuRusTop}
Wenchang Zhu, Paat Rusmevichientong, and Huseyin Topaloglu.
A unified framework to impose market share constraints for selected product classes: Randomized and deterministic assortments under the multinomial logit model.
\emph{Manufacturing \& Service Operations Management}, 28(1):172--192, 2026.
\end{thebibliography}
\end{document}